\documentclass[12pt]{article}

\usepackage[utf8]{inputenc}
\usepackage[T1]{fontenc}
\usepackage[left=2.5cm,right=2.5cm,top=2.5cm,bottom=3cm]{geometry}
\usepackage{amsmath, amssymb, amsthm, amsfonts} 
\usepackage{bbm} 

\usepackage{graphicx}
\usepackage{float}
\usepackage{makeidx}
\usepackage{hyperref}
\usepackage{verbatim}
\usepackage[font=small,labelfont=bf]{caption}
\usepackage{subfigure}
\usepackage{enumerate}
\usepackage{cancel}
\usepackage{mathtools}
\usepackage{MnSymbol}
\usepackage[skip=10pt plus1pt, indent=0pt]{parskip}

\usepackage{tikz}
\usepackage{array}
\usetikzlibrary{patterns.meta, calc, patterns, math, intersections, scopes, shapes, arrows, automata}
\usepackage{pgfplots}
\usepackage{tcolorbox}
\newtcolorbox{mybox}{colback=red!5!white,colframe=red!75!black}

\newlength{\altezzautile}
\newlength{\larghezzautile}
\newtheorem{theorem}{Theorem}
\newtheorem{remark}{Remark}

\makeindex

 \let\b=\beta

       \let\l=\lambda
  \let\om=\omega     
 \let\s=\sigma  
 \let\x=\xi 
   \let\L=\Lambda 
\let\O=\Omega   \let\Si=\Sigma 
\let\S=\Sigma

\def\GI{\mathbb{G}}
\def\EE{\mathbb{E}}
\def\Z{\mathbb{Z}}

\def\Z2{\mathbb{Z}^2}

\def\PP{\mathcal{P}}

\def\0{\emptyset}

\def\CC{\mathcal{C}}

\def\1{\mathbbm{1}} 
\def\<{\langle}
\def\>{\rangle}
\def\be{\begin{equation}}
\def\ee{\end{equation}}
\def\PP{{\mathbb P}}

\def\XX{{\rm X}}
\def\YY{{\rm Y}}

\begin{document}

\title {The BEG model at the FAD triple point on the square lattice}
\author{Estevão F. Borel$^{\dag}$,   Aldo Procacci$^\dag$, Rémy Sanchis$^\dag$ and Benedetto Scoppola$^*$ \\\\
\footnotesize{$^\dag$ Dep. Matem\'atica-ICEx, UFMG, CP 702 Belo Horizonte - MG, 30161-970 Brazil}\\
\footnotesize{$^*$Dipartimento di Matematica - Universita Tor Vergata di Roma, 00133 Roma, Italy}\\
}
\date{\today}
\maketitle
\begin{abstract}
In this note, we prove that the two-dimensional Blume-Emery-Griffiths model at the triple point Ferromagnetic-Antiquadrupolar-Disordered (FAD) has a unique Gibbs measure at any temperature, thereby establishing the absence of phase transitions. The FAD point lies at the intersection of lines separating three regions of the phase diagram, and it is a singular point where the model exhibits infinitely many ground states. Our proof is based on a random-cluster type representation with configuration-dependent cluster weights and a coupling with Bernoulli site percolation with parameter $1/2$.
\end{abstract}

\subsection*{Introduction}
The Blume-Emery-Griffiths model \cite{beg71} defines a class of spin models in the $d$-dimensional unit cubic lattice $\mathbb{Z}^d$ depending on two real parameters $X$ and $Y$ with formal Hamiltonian
\begin{eqnarray}\label{defbeg}
 H (\sigma)=-\sum_{x\sim y} (\sigma_x\sigma_y+\YY\sigma_x^2\sigma_y^2+\XX(\sigma_x^2+\sigma_y^2)),
\end{eqnarray}
where by the notation  $x\sim y$ we mean that $\{x,y\}$ is an unordered pair of nearest neighbors in $\mathbb{Z}^d$, and $\sigma_x\in\{-1,0,+1\}$. It was  introduced in 1971 with the initial motivation to investigate
phase separation and
superfluidity in mixtures of liquid $^3$He and $^4$He and has since been used in many other applications (see e.g. \cite{LMPS} and references there). The parameter space $X-Y$ is divided into three regions (according to the ground states). Namely,  the disordered region $D = \{(\XX, \YY) : 1 + 2\XX + \YY < 0 \mbox{ and }  \XX < 0\}$ where the only ground state is $\s_x=0$ everywhere, the ferromagnetic region $F = \{(\XX, \YY) : 1 + 2\XX + \YY > 0 \mbox{ and }   1 + \XX + \YY > 0\}$ with two ground states, $\s_x=+1$ everywhere and $\s_x=-1$ everywhere, and the antiquadrupolar region $A = \{(\XX, \YY) : 1 + \XX + \YY < 0 \mbox{ and }  \XX > 0\}$ where the model exhibits an infinite  number of ground states  organized  into two equivalent classes (the first class is  such that $\sigma_x=0$ for all  $x$ in the even sublattice  and
$\sigma_x\neq 0$ for all  $x$ in the odd sublattice, and the second class is the opposite). These  three regions are separated by three lines, namely the line
$DF=\{(\XX,\YY)\,:\, 1+2\XX+\YY=0 \mbox{ and } \XX<0\}$ (disordered-ferromagnetic line), the line
$AF =\{(\XX,\YY)\,: \, \XX+\YY+1=0 \mbox{ and }  \XX>0\}$ (antiquadrupolar-ferromagnetic line), and
$AD =\{(0,\YY)\,:\, \YY<-1\}$  (antiquadrupolar-disordered  line), which meet at the ferromagnetic-antiquadrupolar-disordered ($FAD$)  point $\{(0,-1)\}$.

In the regions $F$ and $D$ and at the interface $DF$, the low-temperature regime  of the model can be described by the usual Pirogov-Sinai theory \cite{si82,sv17} while at any point of the $X-Y$ plane  the high temperature regime can be  described  by standard methods such as polymer expansion \cite{llp} and Dobrushin uniqueness criterion \cite{pcl21}.  In the region $A$  the extension of the Pirogov-Sinai theory given in  \cite{gs88} can be used to describe the low temperature regime of the model.

Things become considerably more delicate along the AF and AD lines and at the FAD point, where the model exhibits an infinite number of ground states that cannot be organized into a finite number of equivalence classes. In particular, the residual entropy of the BEG model at the FAD point is higher than at any other point along the AD and AF lines (see \cite{bl08,ln11}).
As a consequence of these features, despite the enormous amount of work and the large number of results available in the literature — most of which have been obtained using mean-field theory, Monte Carlo simulations, and renormalization group techniques — rigorous results for these interfaces remain scarce. 

In this article we consider the  BEG model on the two-dimensional unit square lattice at the  $FAD$ point.
At zero temperature, this model coincides with the  Widom-Rowlinson model \cite{GL71} (with activity $\l=1$) for which uniqueness of the Gibbs measure has been established  by  Higuchi
for a range of activities which includes the value   $\l=1$ \cite{Hi83}. An alternative proof of the uniqueness of the (uniform) Gibbs measure
of this model at zero temperature has recently been given in \cite{LMPS} via a coupling of the Gibbs
measure with + boundary condition  with the Bernoulli site percolation process at parameter 1/2. To our knowledge, no rigorous analysis has been available in the literature for the BEG model at the FAD point at non-zero temperatures, so that the possible presence of a phase transition at some positive $\beta$ cannot, in principle, be ruled out. We stress that rigorous proofs of the existence of phase transitions along the AD line and (possibly) at the FAD point are available only for sufficiently large dimensions $d$ (see \cite{PS} and references therein).

The main result of this note is to prove that the
mean value of the spin at the origin with $+$ or $-$ boundary conditions vanishes in the thermodynamic limit at any non-zero temperature.
Since the BEG model at the FAD point falls into the class of models considered by Lebowitz and Monroe \cite{LM}, for which FKG inequalities hold (see Theorem 1 in \cite{LM}), the equality of local magnetizations with $  +  $ and $  -  $ boundary conditions in the thermodynamic limit implies the independence of any correlation function from the boundary conditions — and thus the uniqueness of the Gibbs state at all temperatures. The (standard) proof of this implication for the BEG model at the FAD point can be carried out along the same lines as for the ferromagnetic Ising spin-1/2 model (see, e.g., \cite{sv17}).

Our result is obtained via a coupling between a random cluster model, encoding the dependence on the temperature, and the zero-temperature BEG model at FAD. We were inspired by a similar coupling described in \cite{D-C} and we also made use of results at the zero-temperature case obtained in \cite{LMPS} (in particular Lemma 1 and its consequence there).

\section{The Model and Results}
\label{sec:model}

Let $L$ be an integer and let  $\L$ be a square in $\Z2$ of size $2L+1$ centered at the origin $o$ of $\Z2$. We denote by $\partial_i\L$ the internal boundary of $\L$, by  $\partial_e\L$ the external boundary of $\L$ and we set $\bar\L=\L\cup\partial_e\L$.
The notation $\lim_{\L\uparrow \infty}$  (i.e. the thermodynamic limit) means  here $\lim_{L\to\infty}$.
At each site $x\in \L$ we place a spin $\s_x$ taking values in the set $\{-1, 0, +1\}$ and
we denote by $\Si_\L= \{-1, 0, +1\}^\L$ the set of all possible spin configurations in $\L$. A boundary condition for the model is determined by a fixed configuration $\x\in \Si_{\mathbb{Z}^2}$ with the understanding that in the thermodynamic limit, as $\L$ invades $\mathbb{Z}^2$, the  sites of $\x$ entering $\L$ are disregarded and those in  $\mathbb{Z}^2\setminus\L$ are kept. At the FAD point the Hamiltonian of the BEG model in the volume $\L$ with boundary condition $\x$ is

\begin{equation}
\label{eq:hamiltonian}
H_\Lambda^\xi(\sigma) = -\sum_{\substack{\{x,y\}\subset \L \\ |x-y|=1}} \sigma_x\sigma_y(1-\sigma_x\sigma_y)
- \sum_{x\in \partial_{i}\Lambda} \sum_{\substack{y\in \partial_e\L \\ |x-y|=1}} (\x_y \sigma_x - \x_y^2\sigma_x^2).
\end{equation}
In particular ``$+$'' is the configuration such that $\x_x=+1$ for all $x\in \Z2$, and
``$-$'' is the configuration such that $\x_x=-1$ for all $x\in \Z2$.

The Gibbs measure $\mu_{\L,\b}^\xi(\cdot)$  of the BEG model in the volume $\L$ at the FAD point at inverse temperature $\b$ with $\xi$ boundary conditions, is given by
\[
\mu_{\L,\b}^\xi(\cdot) = \frac{\sum_{\s\in \Si_\L}(\cdot)e^{-\b H_\L^\xi(\s)}}{Z_\L^\xi(\b)}
\]
where
\[
Z_\L^\xi(\b)=\sum_{\s\in \Si_\L}e^{-\b H_\L^\xi(\s)}.
\]
is the partition function. The mean value of the spin at any fixed site $x\in \L$  under the probability measure $\mu_{\L,\b}^\xi(\cdot)$ is given by
\[
\langle\sigma_x\rangle_{\L,\beta}^\xi= \frac{\sum_{\s\in \Si_\L}\s_x e^{-\b H_\L^\xi(\s)}}{Z_\L^\xi(\b)}.
\]
The thermodynamic limit of $\<\s_x\>_{\L,\b}^\xi$, if it exists,  will be denoted hereafter by $\langle\sigma_x\rangle^\xi_\b$, i.e.
\begin{equation}
\langle\sigma_x\rangle^\xi_\b=\lim_{\Lambda\uparrow\infty}\langle\sigma_x\rangle_{\Lambda,\b}^\xi.
\end{equation}

We can now state our main theorem.\\

\begin{theorem}\label{t1}
For the BEG model on $\Z2$ at the  FAD point and at any inverse temperature $\b\in[0,+\infty]$  we have that
\be\label{th1}
\langle\sigma_x\rangle^+_\b=\langle\sigma_x\rangle^-_\b=0.
\ee
\end{theorem}
As noted in the introduction, the validity of the FKG inequalities for this model, combined with Theorem 1, implies uniqueness of the infinite-volume Gibbs measure and hence rules out the existence of a phase transition.

\section{Proof of Theorem \ref{t1}}

In order to prove Theorem \ref{t1}, we develop a graphical representation of the BEG model at the FAD point. Our construction generalizes the standard Random Cluster representation (FK) used for the Potts model (as discussed,  for instance, in  \cite{D-C}). We begin by constructing this representation in Section \ref{sec:RCR}, then use it in Section \ref{sec:MagCon} to relate the magnetization to a connectivity event, and we conclude the proof of Theorem \ref{t1} in Section \ref{sec:Conc}.

\subsection{Random Cluster Representation}\label{sec:RCR}
As usual we look at $\L$ ($\bar\L$) as a vertex set of a graph $\GI_\L$ ($\GI_{\bar\L}$) whose edges are the neighbor pairs.
So a  set $U\subset\L$ is connected if the induced graph $\GI|_U$ is connected. If $U\subset\L$ is connected, we denote by $\EE_U$
 the set of all edges of $\GI|_U$ and we denote by $\CC_\L$ the set of all connected sets in $\L$.
Now, given a configuration $\s\in \Si_\L$ let $\L^+(\s)=\{x\in \L: \s_x=+1\}$ and let
$\L^\pm(\sigma)=\{x\in \L: \s_x\neq 0\}$.
Let
$C_o^+(\s)$ be the connected component of $\L^+(\s)$ containing the origin $o$ ($C_o^+(\s)$ is empty if $\s_o\neq +1$)
and let $C_o^\pm(\s)$ be the connected component of $\L^\pm(\s)$ containing the origin $o$ ($C_o^\pm(\s)$ is empty if $\s_o=0$). From now on we will set for brevity $\EE_{\L}=E^\circ$  and $\EE_{\bar\L}= E$.

Let $e=\{x,y\}$. If both $x,y\in \L$, set $\sigma_e=\sigma_x\sigma_y$; if $x\in \L$ and $y\in \partial_e\L$, set $\s_e=\s_x\x_y$. Therefore, we can write the Hamiltonian, for any configuration $\s\in \S_\L$ and any fixed boundary condition $\x$, as
\[
H^\x_\L(\sigma)=\sum_{e\in E}[(\sigma_e)^2-\sigma_e]=2\sum_{e\in E}\mathbbm{1}_{\{\sigma_e=-1\}}.
\]
Note that, using the identity
$$
e^{-2\beta\mathbbm{1}_{\s_e=-1}}=e^{-2\beta}+(1-e^{-2\beta})\mathbbm{1}_{\{\s_e\neq-1\}},
$$
and setting $p=(1-e^{-2\b})$,  the partition function $Z_\L^\x(\b)$  can be rewritten as:
\begin{equation}
\begin{split}
   Z_\L^\x(\b) & =\sum_{\sigma\in\Sigma_\L}\prod_{\{x,y\}\in E}[(1-p)+p\mathbbm{1}_{\{\sigma_e\neq -1\}}].
\end{split}
\end{equation}
We introduce a graphical representation by considering the space of bond configurations $\Omega_{E}=\{0,1\}^{E}$.
Given $\om\in \O_E$,
let
$E_\omega=\{e\in E:~\omega_e=1\}$ and $\bar\L_\om=\{x\in \bar\L: x\in e\;\mbox{for some $e\in E_\om$}\}$. We set
$\GI_\om=(\bar\L,E_\om)$. We also set $\partial\L_e^\om=\{x\in \partial\L_e: ~\mbox{$x\in e$ for some $e\in E_\omega$}\}$.

For a fixed boundary condition $\x$, we then define a coupling on the product space $\Sigma_\L\times \Omega_{E}$ via a compatibility relation. We say a spin configuration $\sigma\in \S_\L$ and a bond configuration $\omega\in \Omega_{E}$ are compatible, and write $\s\overset{_\x}{\sim}\om$, if for every edge $e=\{x,y\}\in E$:
\[
\omega_e=1\implies \sigma_e\neq -1.
\]
We then define the joint probability measure $P^\x_p$ on $\Sigma_\L\times \Omega_{E}$ as:
\begin{equation}
     P^\x_p(\sigma,\omega):=\frac{1}{\Xi^\x_{\L}(p)}\prod_{e\in E}[(1-p)
    \delta_{\omega_e,0}+p\delta_{\omega_e,1}\mathbbm{1}_{\{\sigma_e\neq-1\}}].
\end{equation}
Where $\Xi^\x_{\L}(p)$ is the natural normalizing constant. However, notice that
\begin{equation}
\begin{split}
   \Xi^\x_{\L}(p)& =\sum_{\sigma\in\Sigma_\L}\sum_{\omega\in \Omega_{E}}
   \prod_{e\in E}[(1-p)\delta_{\omega_e,0}+p\delta_{\omega_e,1}\mathbbm{1}_{\{\sigma_e\neq-1\}}]\\
    & =\sum_{\sigma\in\Sigma_\L}\prod_{e\in E}\Big[
    \sum_{\omega_e\in\{0,1\}}[(1-p)\delta_{\omega_e,0}+p\delta_{\omega_e,1}\mathbbm{1}_{\{\sigma_e\neq-1\}}]\Big]\\
    & = \sum_{\sigma\in\Sigma_\L}\prod_{e\in E}[(1-p)+p\mathbbm{1}_{\{\sigma_e\neq -1\}}]\\& = Z_\L^\x(\b).
\end{split}
\end{equation}
\begin{remark}
Note that if $(\s,\omega)\in \Sigma_\L\times \Omega_{E}$ is not compatible then $P_p^\x(\sigma,\omega)=0$. Indeed, if  $(\s,\omega)$ is not compatible, then there  is at least  an edge $e\in E$ such that $\omega_e=1$ and $\s_e=-1$. But for such an edge
$$
(1-p)
    \delta_{\omega_e,0}+p\delta_{\omega_e,1}\mathbbm{1}_{\{\sigma_e\neq-1\}}=0.
$$
\end{remark}
Summing $ P^\x_p(\sigma,\omega)$ over $\sigma$ yields the marginal distribution on $\Omega_{E}$, denoted by $\phi^\xi_p$. To calculate $\phi^\xi_p(\omega)$, observe that
$$
\prod_{e\in E}[(1-p)
    \delta_{\omega_e,0}+p\delta_{\omega_e,1}\mathbbm{1}_{\{\sigma_e\neq-1\}}]=(1-p)^{|E|-|\omega|}p^{|\omega|}\prod_{e\in E_\om}\mathbbm{1}_{\{\sigma_e\neq-1\}},
$$
where $|\omega|$ denotes the number of open edges of the configuration $\omega\in\Omega_{E}$ (i.e. those edges $e\in E$ such that $\omega_e=1$). Hence
\begin{equation}
    \phi^\x_p(\omega)=\sum_{\sigma\in \Sigma_\L}P^\x_p(\sigma,\omega)=\frac{1}{ Z^\x_{\L}(\b)}p^{|\omega|}(1-p)^{|E|-|\omega|}W^\x(\omega),
\end{equation}
where
\[
W^\x(\omega)=\sum_{\sigma \in \Sigma_\L} \prod_{e\in E_\om}\mathbbm{1}_{\{\sigma_e \neq -1\}}=
\sum_{\sigma\in \Sigma_\L}\mathbbm{1}_{\{\sigma\overset{_\x}{\sim} \omega\}}.
\]
To understand what the factor  $W^\x(\omega)$ represents, let us introduce some notation. Given a configuration
$\omega\in \Omega_{E}$ and a boundary condition $\x$ on $\partial_e\L$, let $\x^\om$ be defined as follows.
\begin{equation}
\x^\om_x=\begin{cases}\x_x & \mbox{if $x\in \partial\L_e^\om$}\\0 & {\rm otherwise}
\end{cases}.
\end{equation}
Let now $E^\circ_\omega\subset E_\om$ be such that $E^\circ_\omega=\{e\in E^\circ:~\omega_e=1\}$ and let $\GI^\circ_\omega=(\L,E^\circ_\omega)$ (recall that $E^\circ$ is the set of neighbor pairs in $\L$ while $E$ is the set of neighbor pairs in $\L\cup\partial_e\L$).
Then
$$
W^\x(\omega)=Z^{\x}_{\GI^\circ_\omega}(\beta=+\infty)
$$
where $Z^{\x}_{G^\circ_\omega}(\beta=+\infty)$ is the partition function  of the zero-temperature BEG model at
FAD on the graph $\GI^\circ_\omega$ with $\x^\om$ boundary conditions.

Now, given $\omega\in \Omega_E$ we define the (conditional) measure $P^\xi(\s\mid\om)$ on $\S_\L$ in a natural way as follows.
\begin{equation}
P^\xi(\s\mid\omega)= \frac{P^\x_p(\sigma,\omega)}{\phi^\x_p(\omega)}=\frac{\prod_{e\in E_\om}\mathbbm{1}_{\{\sigma_e\neq-1\}}}{W^\x(\omega)} = \frac{\mathbbm{1}_{\{\sigma\overset{_\x}{\sim} \omega\}}}{W^\x(\omega)}.
\end{equation}

Finally, since
\be\nonumber
\begin{split}
\sum_{\om\in \O_{E}}P^\x_p(\sigma,\omega) & =\frac{1}{Z^\x_{\L}(\b)}\sum_{\om\in \O_{E}}\prod_{e\in E}[(1-p)
    \delta_{\omega_e,0}+p\delta_{\omega_e,1}\mathbbm{1}_{\{\sigma_e\neq-1\}}]\\
    & =\frac{1}{Z^\x_{\L}(\b)}\prod_{e\in E}[(1-p)
    +p\mathbbm{1}_{\{\sigma_e\neq-1\}}] = \mu^\x_{\L,\b}(\sigma),
\end{split}
\ee
 for any function $g:\Sigma_\L\to \mathbb{R}$ we have the decomposition:
\begin{equation}\label{eq:entropy_decomposition}
\begin{split}
    \langle g \rangle^\x_{\L,\b}
     & =\sum_{\s\in \S_\L} g(\s)\mu^\x_{\L,\b}(\sigma)\\
    & =\sum_{\om\in \O_{E}} \sum_{\s\in \S_\L} g(\sigma)P^\x_p(\sigma,\omega)\\
    & =\sum_{\om\in \O_{E}} \sum_{\s\in \S_\L} g(\sigma)P^\xi(\sigma\mid\omega)\phi^\x_p(\omega)\\
    & =\sum_{\om\in \O_{E}}\phi^\x_p(\omega) \sum_{\s\in \S_\L} g(\sigma)\Bigg\{\frac{\mathbbm{1}_{\{\sigma\sim\omega\}}}{W^\x(\omega)}\Bigg\}.
    \end{split}
\end{equation}
I.e. in the end we get
\begin{equation}\label{mean}
 \langle g \rangle^\x_{\L,\b} =\EE_{\phi^\x_p}[ \<g(\s)\>_{\GI_\om,\b=+\infty}^\xi].
\end{equation}

Identity (\ref{mean}) shows clearly the utility of this coupling which allows us to separate the  thermal and geometric effects. In the r.h.s. of (\ref{mean}), the measure $\phi^\xi_p$ encodes the temperature dependence via the parameter $p=(1-e^{-2\b})$,  while, for fixed $\om\in \O_{E}$, $\<g(\s)\>_{\GI_\om,\b=+\infty}^\xi$ is the mean value of $g(\s)$ for the zero-temperature BEG model at FAD on the graph $\GI_\om$ with $\x^\omega$ boundary conditions.

\subsection*{Magnetization-Connectivity Identity}\label{sec:MagCon}

We now apply identity \eqref{mean} to express the magnetization at a given site in terms of a connectivity property under the zero-temperature measure. To this end let us introduce some notation. For a fixed bond configuration $\omega\in\Omega_E$ with boundary condition $\xi$, let $\bar{\mu}^\x_{\omega}$ denote the uniform measure  on the set of spin configurations $\sigma\in\S_\L$ compatible with $\omega$. In other words,  $\bar{\mu}^\x_{\omega}$ is the Gibbs measure
 of the zero-temperature BEG model at FAD on the graph $\GI_\om$ with $\x^\om$ boundary conditions.
Moreover, for a given $x\in \L$,
let $C_x(\omega)\subset \L$ denote the vertex set of the connected component of the graph $\GI_\om$ containing  $x$.  Then we have that $\O_{E}=\O_{E}^s\uplus\O^c_{E}$ where $\uplus$ denotes the disjoint union and $\O_{E}^s=\{\om\in \O_{E}: ~ C_x(\om)\cap\partial_e \L=\0\}$ and $\O_{E}^c=\{\om\in \O_{E}: ~ C_x(\om)\cap\partial_e \L\neq \0\}$ for each $x\in\L$.

Let us fix a configuration $\omega$ and let us choose the boundary condition $+$ and a vertex $x\in \L$. Then we define the positive connectivity event $$\{x \xleftrightarrow[\omega]{+} \partial \L\}$$ (event in $\Sigma_\L$) as the set of spin configurations such that there is a path $p$ in the graph $\GI_\om$ with vertex set $V_p\subset \bar\L$ connecting  the vertex $x$ to some vertex of the boundary $\partial_e \L$ with $\sigma_y=+1$ for all $y\in V_p$.\\

\begin{remark}
Note that when $\om\in \O_{E}^s $ (i.e when $C_x(\om)$ does not touch the boundary) then, by definition $\bar{\mu}^+_\omega(x \xleftrightarrow[\omega]{+} \partial\L)=0$. This is despite the fact that  for some configurations $\s\in \S_\L$ compatible with $\om$ there may well be a path starting at $x$  and ending at a vertex $y\in \partial_e\L$ all made by $+$ spins.   However, since $ \om\in \O_{E}^s$,  such a path is not a path in $\GI_\om$ and thus
such  configurations do not belong to the set
$\{x \xleftrightarrow[\omega]{+} \partial \L\}$. In other words, the set $\{x \xleftrightarrow[\omega]{+} \partial \L\}$
is empty when $ \om\in \O_{E}^s $.\\
\end{remark}

\begin{theorem}[Magnetization-Connectivity Identity]\label{thm:mag_con}
For the BEG model at the FAD point, the finite-temperature magnetization at site $x$ under positive boundary conditions satisfies:
    \begin{equation}
        \langle \sigma_x\rangle_{\L,\beta}^+=\mathbb{E}_{\phi^+_p}\Big[\bar{\mu}^+_\omega(x \xleftrightarrow[\omega]{+} \partial\L)\Big].
    \end{equation}
\end{theorem}

\begin{proof} Using the decomposition established in (\ref{eq:entropy_decomposition}) and (\ref{mean}), we have initially:
\[
\langle \sigma_x\rangle_{\L,\beta}^+= \sum_{\omega\in\Omega_E} \frac{1}{W^+(\omega)}\sum_{\sigma\overset{_+}{\sim}\omega}\sigma_x.
\]
\textbf{Case I:} when $\om\in \O_{E}^s$.

Consider a configuration $\omega\in \O_{E}^s$ (i.e. the cluster $C_x(\omega)$ does not reach the boundary). In this scenario, the spins belonging to $C_x(\omega)$ are decoupled from the boundary.
Let $\sigma\overset{_+}{\sim}\omega$ and consider a transformation $\sigma \to \sigma'$ that flips all spins inside $C_x(\om)$ while leaving unchanged all spins in $\L\setminus C_x(\om)$: namely,  $\sigma'_u = -\sigma_u$ if $u \in C_x(\om)$, and $\sigma'_u = \sigma_u$ otherwise. Clearly we have that  $\sigma'\overset{_+}{\sim}\omega$. By symmetry, the contributions of $\sigma$ and $\sigma'$  in the sum $\sum_{\sigma\overset{_+}{\sim}\omega}\s_x$ cancel out, yielding
\[
\frac{1}{W^+(\omega)}\sum_{\sigma\overset{_+}{\sim}\omega}\sigma_x=0.
\]
\textbf{Case II:} when $\om\in \O_{E}^c$.

Let now assume that $\omega\in \O_{E}^c$.  Since the boundary condition is fixed to $+$, in a configuration $\s\overset{_+}{\sim}\om$, any path of non-zero spins in $\GI_\om$ starting at some vertex of $\partial_e\L$ must consist entirely of $+1$ spins. We decompose the sum $\sum_{\s\overset{_+}{\sim}\om}\sigma_x$ as follows:
\[
\sum_{\s\overset{_+}{\sim}\om} \sigma_x = \sum_{\substack{\s\overset{_+}{\sim}\om\\ \sigma_x = +1}} 1 -
\sum_{\substack{\s\overset{_+}{\sim}\om \\ \sigma_x = -1}} 1.
\]
Set for brevity
$\S_\L^{\om, \pm}=\{\s\in \S_\L: \s_x=\pm 1,\s\overset{_+}{\sim}\om\}$. Observe that the set  $\S_\L^{\om,+}$ can be split into two disjoint events.
\begin{enumerate}
    \item $E_c =\{\s\in \S_\L^{\om,+}: \mbox{ a path in $C_x(\om)$ of $+$ spins connects $x$ to the boundary $\partial_e\L$} \}$.
    \item $E_s =\{\s\in \S_\L^{\omega,+}: \mbox{ no path in $C_x(\om)$ of $+$ spins containing  $x$ reaches  the boundary  $\partial_e\L$}\}.$
\end{enumerate}
Note that if $\s\in E_s$ then the vertex $x$ has spin $+1$ but is separated from the boundary by zeros within $C_x(\om)$.

Similarly, if $\s\in \S_\L^{\omega,-}$, and thus $\s_x=-1$, no path in $C_x(\om)$ of spins $-1$ containing  $x$ reaches  the boundary  $\partial_e\L$. In other words  $x$ resides in a finite component of $-$  spins   surrounded by zeros inside $C_x(\om)$ (even though $C_x(\om)$ itself touches the boundary).

Let $K_x(\sigma) \subset C_x(\om)$ be the connected component of non-zero spins containing $x$.
If $\sigma \in E_s$ or if $\sigma\in \S_\L^{\omega,-}$, then $K_x(\sigma)$ does not touch $\partial_e\L$ and is bounded by zeros inside $C_x(\om)$. By the same spin-flip symmetry argument used in Case I, there is a bijection between configurations in $E_s$ (where $K_x(\s)$ is all $+$) and configurations in $\S_\L^{\omega,-}$ (where $K_x(\s)$ is all $-$).

Consequently, their contributions to the magnetization cancel out:
\[
\sum_{\sigma \in E_s} 1 - \sum_{\sigma \in\S_\L^{\omega,-}} 1 = 0.
\]
Hence we get
\[
\sum_{\sigma \in \S_\L^{\omega,+}} \sigma_x  =
\sum_{\s\in E_c}1
= |\{\sigma \in \S_\L^{\omega,+} : x \xleftrightarrow[\omega]{+} \partial\L\}|.
\]
Dividing by $W^+(\omega)$, we conclude:
\begin{equation}
\frac{1}{W^+(\omega)}\sum_{\sigma\overset{_+}{\sim}\omega}\sigma_x~= ~\begin{cases}
\displaystyle{\frac{|\{\sigma \in  \S_\L^{\omega,+} : x \xleftrightarrow[\omega]{+} \partial \L\}|}{W^+(\omega)}} & \mbox{if $\om\in \O_{E}^c$}\\\\
0& \mbox{if $\om\in \O_{E}^s$}
\end{cases}
~= ~\bar{\mu}^+_\omega(x \xleftrightarrow[\omega]{+} \partial\L).
\end{equation}
\end{proof}


%

\subsection{Conclusion of the proof of Theorem \ref{t1}: the vanishing of the magnetization}\label{sec:Conc}

We now have all the ingredients needed to conclude the proof of Theorem \ref{t1}. The results given in \cite{LMPS} imply that
\be\label{lmps}
\bar{\mu}^+_{\omega}(x \xleftrightarrow[\omega]{+} \partial \L)\le \PP_{\GI_\om}^{1/2}(x \leftrightarrow \partial \L),
\ee
where $\PP_{\GI_\om}^{1/2}(x \leftrightarrow \partial \L)$ is the probability of the event that $x$ is connected to the boundary $\partial_i\L$
through a path of open sites inside the graph $\GI_\om$ in the Bernoulli site percolation on the graph $\GI_\om$.
Indeed, as discussed in  Section 2.2 of \cite{LMPS}, the inequality (\ref{lmps}) can be proven via a coupling of two Markov chains, one on the zero-temperature BEG model at FAD on the graph $\GI_\om$ and the other on the Bernoulli site percolation on the same graph.

Now, since clearly
$\PP_{\GI_\om}^{1/2}(x \leftrightarrow \partial \L)\le \PP_{\L}^{1/2}(x \leftrightarrow \partial \L)$
we get that
$$
\bar{\mu}^+_{\omega}(x \xleftrightarrow[\omega]{+} \partial \L)\le \PP_{\L}^{1/2}(x \leftrightarrow \partial \L)
$$
and thus
$$
\langle \sigma_x\rangle_{\L,\beta}^+= \mathbb{E}_{\phi^+_p}[\bar{\mu}^+_{\omega}(x \xleftrightarrow[\omega]{+} \partial \L)]\le \mathbb{E}_{\phi^+_p}[\PP_{\L}^{1/2}(x \leftrightarrow \partial \L)]=\PP_{\L}^{1/2}(x \leftrightarrow \partial \L)
$$
and hence
$$
0\leq\lim_{\L\uparrow\infty} \langle \sigma_x\rangle_{\L,\beta}^+\le \lim_{\L\uparrow\infty} \PP_{\L}^{1/2}(x \leftrightarrow \partial \L)=0.
$$

\section*{Conclusions}

We introduced a random-cluster type representation for the BEG model at the FAD point and used it to establish the absence of phase transition on the square lattice. We remark that our argument applies beyond the square lattice $\mathbb{Z}^2$. This result also extends to any infinite graph $G$ for which the critical parameter for Bernoulli site percolation satisfies $p_c(site,G)>1/2$, since the stochastic domination by a Bernoulli site percolation with parameter $1/2$ holds independently of the graph structure.

For $\mathbb{{Z}}^d$ with $d\geq 3$, the situation is fundamentally different. Since $p_c(site,\mathbb{Z}^d)<1/2$ for all $d\geq3$, the upper bound provided by Bernoulli site percolation at parameter $1/2$ becomes supercritical and no longer forces the magnetization to vanish. In this regime, Monte Carlo simulations \cite{LMPS} suggest that the zero-temperature magnetization at $d=3$ is strictly positive, and it is expected that a non-trivial critical inverse temperature $\b_c\in(0,\infty)$ should exist.

The non-uniqueness of Gibbs measures at zero temperature has only been proved when $d$ is sufficiently high \cite{PS}. Proving this for intermediate dimensions, however, remains a challenging open problem. A natural approach would be to establish a stochastic domination from below by a supercritical Bernoulli percolation process. However, classical local-comparison tools such as Holley's inequality are not useful in this direction at the zero-temperature setting. Another important factor is the well-known absence of monotonicity properties for hard-core models \cite{haggstrom02}. New ideas seem to be required, and a deeper understanding of the properties of the measure $\phi_p$ and the zero-temperature system in higher dimensions could be a promising direction to understand the model at positive temperatures.

\subsubsection*{Acknowledgements}
E. F. B. was partially supported by CAPES and FAPEMIG, R. S. was partially supported by CNPq, CAPES and FAPEMIG (grant RED-00133-21) and BS acknowledges the support of the Italian MIUR Department of Excellence Grant (CUP E83C23000330006). E. F. B. thanks the Department of Mathematics of Università degli Studi di Roma Tor Vergata for the hospitality and for providing a productive research environment during his research visit.

\subsubsection*{Competing Interests Statement}
The authors declare none.
\subsubsection*{Data Availability Statement}
No data was used for the research described in the article.
\subsubsection*{Funding statement}
This work was supported by CAPES and FAPEMIG (grant RED-00133-21) and Italian MIUR Department of Excellence Grant (CUP E83C23000330006). The funder had no role in study design, data collection and analysis, decision to publish, or preparation of the manuscript.

\end{document}